\documentclass[11pt]{article}
\def\nottoobig#1{{\hbox{$\left#1\vcenter to1.111\ht\strutbox{}\right.\n@space$}}}

\newtheorem{fact}{Fact}

\newtheorem{theorem}{Theorem}[section]

\newtheorem{lemma}[theorem]{Lemma}

\newtheorem{definition}[theorem]{Definition}

\usepackage{epsfig}
\usepackage{amsmath}
\usepackage{amsfonts}

\newcommand{\nat}{{\mathbb N}}

\newcommand{\poly}{{\rm poly}}

\def\nottoobig#1{{\hbox{$\left#1\vcenter
to1.111\ht\strutbox{}\right.\n@space$}}}

\newcommand{\prob}{{\rm Prob}}

\newcommand{\ie}{$\mbox{i.e.}$}

\newlength{\filength}
\newsavebox{\gcbox}
\sbox{\gcbox}{\framebox[\filength]{\rule{0ex}{2ex}}}

\newcommand{\qedblob}{\mbox{\rule[-1.5pt]{5pt}{10.5pt}}}
\def\literalqed{{\ \nolinebreak\hfill\mbox{\qedblob\quad}}}

\def\qed{\literalqed}

\newcommand{\singlespacing}{\let\CS=
\@currsize\renewcommand{\baselinestretch}{1}\tiny\CS}
\newcommand{\singlespacingplus}{\let\CS=
\@currsize\renewcommand{\baselinestretch}{1.25}\tiny\CS}
\newcommand{\doublespacing}{\let\CS=
\@currsize\renewcommand{\baselinestretch}{1.75}\tiny\CS}
\newcommand{\draftspacing}{\let\CS=
\@currsize\renewcommand{\baselinestretch}{2.0}\tiny\CS}

\def\zo{\{0,1\}}

\def\mapping{\rightarrow}

\usepackage{amsfonts}

\if01
\fi

\if
\fi

\makeatletter
\def\@listI{\leftmargin\leftmargini \parsep 4.5pt plus 1pt minus 1pt\topsep6pt plus 2pt minus 2pt \itemsep  2pt plus 2pt minus 1pt}

\let\@listi\@listI
\@listi
\makeatother

\author{ {Marius Zimand\/}
\thanks{  Department of Computer and Information Sciences, Towson University,
Baltimore, MD.; email: mzimand@towson.edu; http://triton.towson.edu/\~{ }mzimand.
The author is supported in part
by NSF grant CCF 1016158.
This is a part of the conference paper~\cite{zim:c:extadvice}. It corrects an erroneus result from there, namely Theorem 4.8 from~\cite{zim:c:extadvice} (see footnote in Fact~\ref{t:boundspolysizeinformal}).}}

\date{ }

\title{Nonuniform Kolmogorov Extractors}

\begin{document}

\maketitle

\begin{abstract}
We  establish tight bounds on the amount on nonuniformity that is necessary for extracting a string with randomness rate $1$ from a single source of randomness with lower randomness rate. More precisely, as instantiations of more general results, we show that
	while $O(1)$ amount of advice regarding the source is not enough for extracting a string with randomness rate $1$ from a source string with constant subunitary random rate, $\omega(1)$ amount of advice is.
\end{abstract}

{\bf Keywords:} random strings, randomness extraction, Kolmogorov extractors.
\smallskip

\section{Introduction}
By and large, randomness extraction is an algorithmical process that constructs a source of randomness of high quality from one or several sources of lower quality. If we restrict to the case when there is only one input source, one wants to design an effective transformation $E$ from the set of $n$-bit strings to the set of $m$-bit strings such that for any source $x$ with randomness at least $k$ (where $k < n$), $E(x)$ has randomness $\approx m$. It is desirable to have $m \approx k$ (\ie, to extract all, or almost all, of the randomness in the source). The problem of randomness extraction has been modeled in two ways. In the first model (which has been studied extensively in computational complexity theory), a source is a probability distribution $X$ over $\zo^n$ and its randomness is given by the min-entropy $H_\infty(X)$. In the second model, a source is a string $x \in \zo^n$, its randomness is given by its Kolmogorov complexity $C(x)$, and the algorithmical procedure is called a Kolmogorov extractor. In this paper, the focus is on Kolmogorov extractors. 

Thus, Kolmogorov extractors are procedures that increase the Kolmogorov complexity rate of strings and sequences. Their explicit study was initiated by Fortnow, Hitchcock, A.~Pavan, Vinodchandran and Wang~\cite{fhpvw:c:extractKol} for the case of finite strings and by Reimann~\cite{rei:t:thesis} for the case of infinite strings. The recent paper~\cite{zim:j:kolmextractsurvey} is a survey of this field.

 It is well-known that randomness extraction by \emph{uniform} procedures from a single source is not possible (if we exclude some trivial cases). We investigate the amount of \emph{non-uniformity} that is necessary for Kolmogorov extractors that extract from one source. As a consequence of a result of Vereshchagin and Vyugin~\cite{ver-vyu:j:kolm}, we note that obtaining a source with randomness rate $1$ from a source with randomness rate, say, $0.99$ is not possible even if the extractor has access to a constant amount of non-uniform information. In contrast, we show that an $\omega(1)$ amount of non-uniform information is sufficient for this task.

We continue with a more detailed discussion of the two types of results and of the technical method that we use.
\smallskip

Formally, given the parameter $k \leq n$, one would like to have a function $E : \zo^n \mapping \zo^m$ such that, whenever $C(x) \geq k$, it holds that $C(E(x)) \approx m$ . It is well known that no such computable function $E$ exists for non-trivial parameters. Indeed for any given $E$, consider the string $y \in \zo^m$ with the largest number of preimages. Then $C(y \mid n) = O(1)$ and among its at least $2^{n-m}$ preimages there must be some $x$ with $C(x) \geq n-m$. In other words, for any given $E$, there are some strings (such as the above $x$) on which $E$ fails. Thus, in order for a function $E$ to extract randomness from any source $x$ with randomness $\geq k$, $E$ must have some additional information $\alpha_x$, which we call \emph{advice about the source.} The question is how much such advice information should be provided.

Fortnow et al.~\cite{fhpvw:c:extractKol} have shown that a constant number of advice bits are sufficient if one settles to extracting from strings with linear randomness a string whose randomness rate is $1-\epsilon$.~\footnote{The randomness rate of an $n$-bit string $x$ is $C(x)/n$.} More precisely, they show that for any positive rational numbers $\sigma$ and $\epsilon$, there exists a polynomial-time computable function $E$ and a constant $h$ such that for any $x \in \zo^n$ with $C(x) \geq \sigma n$, there exists a string $\alpha_x$ of length $h$ such that $C(E(x, \alpha_x)) \geq (1-\epsilon)m$ and $m \geq cn$, for some constant $c$ that depends on $\sigma$ and $\epsilon$. Note that this result implies that it is possible to construct in polynomial time a list with $2^h$ strings and one of them is guaranteed to have Kolmogorov complexity at least $(1-\epsilon)m$. 

The shortcoming of Fortnow et al's result is that the randomness rate of the output is not $1$. It would be desirable that $C(E(x, \alpha_x)) \geq m - o(m)$. We first remark that, as a  consequence of a result of Vereshchagin and Vyugin~\cite{ver-vyu:j:kolm}, randomness rate $1$ cannot be obtained with a constant number of bits of advice about the input. Indeed, we show the following.
\begin{fact}
If a computable function $E: \zo^n \times \zo^h \mapping \zo^m$ has the property that for all strings $x$ with $C(x) \geq \sigma n$, it holds that there exists $\alpha_x$ such that $C(E(x, \alpha_x)) \geq (1-\epsilon)m$, then $\epsilon \geq \frac{1-\sigma}{2^{h+1}-1} - o(1)$ (provided that $m = \omega(\log n +h)$).
\end{fact}

In contrast with the above impossibility result, we show that from sources with a linear amount of randomness, one can extract a string with randomness rate $1$ with basically any non-constant amount of advice, such as, for example, the inverse of the Ackerman function. This is an instantiation of the following more general result.
\smallskip

\begin{fact}
\label{t:boundsinformal}
(Informal statement;  see Theorem~\ref{t:extractsmalladvice} for full statement.)
For any $m = m(n)$ computable from $n$ there exists a computable function $E$ with the following property:

For every $n$-bit string $x$ with complexity $\geq m$, there exists a string $\alpha_x$ of length $\omega( \log \frac{n}{m})$ such that $C(E(x,\alpha_x)) = m - o(m)$ and the length of $E(x, \alpha_x)$ is $m$.
\end{fact}
\smallskip

Note that the function $E$ from Fact~\ref{t:boundsinformal} is computable, but no complexity bound is claimed for it.  We can obtain an extractor $E$ computable by a polynomial-size circuit with almost all the properties from 
Theorem~\ref{t:boundsinformal}. Basically the weakening is that the output length $m$ has to be $\leq  (\log n)^c$, for some positive constant $c$. Moreover the polynomial-size circuit is itself computable, in the sense that there exists an algorithm that on input $n$ outputs the description of the circuit that computes $E$.  We call such a circuit an \emph{effectively constructible  circuit}. (Note: Without the restriction that $E$ is computable, the result would be trivial because a polynomial-size circuit computing $E$ could simply have a random string hardwired into it.)
\begin{fact}
\label{t:boundspolysizeinformal}
(Informal statement;  see Theorem~\ref{t:derandextract} for full statement.)
For every constant $c$ and any $m = m(n)$ computable in polynomial time from $n$ such that $m \leq (\log n)^c$ there exists a function $E$, which is computable by a polynomial-size effectively constructible circuit with the following property:

For every $n$-bit string $x$ with complexity $\geq m$, there exists a string $\alpha_x$ of length $\omega( \log \frac{n}{m})$ such that $C(E(x,\alpha_x)) = m - o(m)$ and the length of $E(x, \alpha_x)$ is $m$.\footnote{In the conference version~\cite{zim:c:extadvice}, it was claimed that this result holds for $m(n) \leq cn$, for some fixed constant $c$. The proof contained an error, and we can only prove the result for $m(n) \leq \log^c(n)$.}
\end{fact}
\medskip

{\bf Discussion of technical aspects.}
We present the main ideas in the proofs of Fact~\ref{t:boundsinformal}, and Fact~\ref{t:boundspolysizeinformal}. A Kolmogorov extractor is a computable ensemble of functions $E: \zo^{n_1} \times \zo^{n_2} \mapping \zo^m$ such that for all $x \in \zo^{n_1}$ and $y \in \zo^{n_2}$ that have Kolmogorov complexity above a certain threshold value and that are sufficiently independent (which roughly means that $C(y \mid x) \approx C(y)$), it holds that $C(E(x,y)) \approx m$.

To obtain the extractors $E$ that require a small amount of non-uniform information about the source (see Fact~\ref{t:boundsinformal} and Fact~\ref{t:boundspolysizeinformal}), we need to show that for each $x$, it is enough to have a short string $\alpha_x$ such that $E(x, \alpha_x)$ has randomness rate $1$ and contains almost all the randomness of $x$. The solution is based on the fact that from $x$ and a short string that is random even conditioned by $x$, one can extract almost all the randomness of $x$. This is similar to the well-studied case of seeded extractors in computational complexity, with the remark that we can have shorter seeds because requiring that the output has randomness rate equal to $1$ is weaker than requiring that the output is statistically close to the uniform distribution (as stipulated in the definition of seeded extractors). Then we take $\alpha_x$ to be such a short seed. The above fact is obtained via an elementary use of the probabilistic method. We first identify a combinatorial object, called a \emph{balanced table}, that characterizes a Kolmogorov extractor, in the sense that the table of a Kolmogorov extractor must satisfy the combinatorial constraints of a balanced table. We show (with the probabilistic method) that such an object exists  with a seed of length $\omega(\log (n/m))$, where $n$ is the length of $x$ and $m$ is the Kolmogorov complexity of $x$. This establishes Fact~\ref{t:boundsinformal}. Since the function $E$ from Fact~\ref{t:boundsinformal} is obtained via the probabilistic method, we cannot claim any complexity bound for it.  To obtain the Kolmogorov extractor in Fact~\ref{t:boundspolysizeinformal}, which is computed by polynomial-size circuits,  we derandomize the construction from Fact~\ref{t:boundsinformal}, using a method of Musatov~\cite{mus:t:spacekolm}. The key observation is that the combinatorial constraints of a balanced table can be checked by constant-depth circuits of relatively small size. The argument goes as follows: (a) these constraints require that in all sufficiently large rectangles of the table no element appears too many times; (b) thus one needs to count the occurrence of each element in every sufficiently large rectangle of the table; (c) by a well-known result of Ajtai~\cite{ajt:j:constantdepthcount},  this operation can be done with sufficient accuracy by constant-depth circuits with relatively small size. Therefore, we can use the Nisan-Wigderson~(\cite{nis-wig:j:hard}) pseudo-random generator NW-gen that fools bounded-size constant-depth circuits and has seeds of size polylogaritmic in the size of the output.  Since balanced tables with the required parameters are abundant, we infer that there exists a seed $s$ so that NW-gen$(s)$ is a balanced table with the required parameters. A balanced table is an object of size exponential in $n$, which implies that the seed $s$ has size polynomial in $n$. Moreover, the Nisan-Wigderson pseudo-random generator has the property that each bit of the output can be calculated separately in time polynomial in the length of the seed. This implies that the Kolmogorov extractor whose table is NW-gen$(s)$ can be computed by a polynomial-sized circuit that has $s$ hard-wired in its circuitry.

\section{Preliminaries}
\subsection{Notation and basic facts on Kolmogorov complexity}
The Kolmogorov complexity of a string $x$ is the length of the shortest effective description of $x$. There are several versions of this notion. We use here  the \emph{plain complexity}, denoted $C(x)$, and also the \emph{conditional plain complexity} of a string $x$ given a string $y$, denoted $C(x \mid y)$, which is the length of the shortest effective description of $x$ given $y$. The formal definitions are as follows.
We work over the binary alphabet $\zo$. A string is an element of $\{0,1\}^*$.
If $x$ is a string, $|x|$ denotes its length.  
Let $M$ be a Turing machine that takes two input strings and outputs one string. For any strings $x$ and $y$, define the \emph{Kolmogorov complexity} of $x$ conditioned by $y$ with respect to $M$, as 
$C_M(x \mid y) = \min \{ |p| \mid M(p,y) = x \}$.
There is a universal Turing machine $U$ with the following property: For every machine $M$ there is a constant $c_M$ such that for all $x$, $C_U(x \mid y) \leq C_M(x \mid y) + c_M$.
We fix such a universal machine $U$ and dropping the subscript, we write $C(x \mid y)$ instead of $C_U(x \mid y)$. We also write $C(x)$ instead of $C(x \mid \lambda)$ (where $\lambda$ is the empty string). The \emph{randomness rate} of a string $x$ is defined as ${\rm rate}(x) = \frac{C(x)}{|x|}$.  If $n$ is a natural number, $C(n)$ denotes the Kolmogorov complexity of the binary representation of $n$. For two $n$-bit strings $x$ and $y$, the information in $x$ about $y$ is denoted $I(x : y)$ and is defined as $I(x : y) = C(y \mid n) - C(y \mid x)$.

In this paper, the constant hidden in the $O(\cdot)$ notation only depends on the universal Turing machine.

For all $n$ and $k \leq n$, 

$2^{k-O(1)} < |\{x \in \zo^n \mid C(x\mid~n) < k\}| < 2^k$.


Strings $x_1, x_2, \ldots, x_k$ can be encoded in a self-delimiting way (\ie, an encoding from which each string can be retrieved) using $|x_1| + |x_2| + \ldots + |x_k| + 2 \log |x_2| + \ldots + 2 \log |x_k| + O(k)$ bits. For example, $x_1$ and $x_2$ can be encoded as $\overline{(bin (|x_2|)} 01 x_1 x_2$, where $bin(n)$ is the binary encoding of the natural number $n$ and, for a string $u = u_1 \ldots u_m$, $\overline{u}$ is the string $u_1 u_1 \ldots u_m u_m$ (\ie, the string $u$ with its bits doubled).

All the Kolmogorov extractors in this paper are ensembles of functions $f = (f_n)_{n \in \nat}$ of type $f_n : \zo^n \times \zo^{k(n)} \mapping \zo^{m(n)}$. For readability, we usually drop the subscript and the expression ``ensemble $f: \zo^n \times \zo^k \mapping \zo^m$'' is a substitute for
``ensemble $f = (f_n)_{n\in \nat}$, where for every $n$, $f_n : \zo^n \times \zo^{k(n)} \mapping \zo^{m(n)}$.''

For any $n \in \nat$, $[n]$ denotes the set $\{1,2, \ldots, n\}$.

\subsection{Approximate counting via polynomial-size constant-depth circuits}
In the derandomization argument used in the proof of Theorem~\ref{t:derandextract}, we need to count with constant-depth polynomial-size circuits. Ajtai~\cite{ajt:j:constantdepthcount} has shown that this can be done with sufficient precision.

\begin{theorem}
\label{t:ajtai}
(Ajtai's approximate counting with polynomial size constant-depth circuits.)
There exists a uniform family of circuits $\{G_n\}_{n \in \nat}$, of polynomial size and constant depth, such that for every $n$, for every $x \in \zo^n$, for every $a \in \{0, \ldots, n-1\}$, and for every $\epsilon > 0$,
\begin{itemize}
	\item If the number of $1$'s in $x$ is $\leq (1 - \epsilon)a$, then $G_n(x,a,1/\epsilon) = 1$,
	\item If the number of $1$'s in $x$ is $\geq (1 + \epsilon)a$, then $G_n(x,a,1/\epsilon) = 0$.
\end{itemize}
\end{theorem}
 We do not need the full strength (namely, the uniformity of $G_n$) of this theorem; the required level of accuracy (just $\epsilon > 0$) can be achieved by non-uniform polynomial-size circuits of depth $d=3$ (with a much easier proof, see~\cite{vio:c:approxcount}).

\subsection{Pseudo-random generator fooling bounded-size constant-depth circuits}
The derandomization in the proof of Theorem~\ref{t:derandextract} is done using the Nisan-Wigderson pseudo-random generator that ``fools" constant-depth circuits~\cite{nis-wig:j:hard}. Typically, it is required that the circuit to be fooled has polynomial size, but the proof works for  circuits of size $2^{\log^c n}$ (where $c$ is any fixed constant).
\begin{theorem}
\label{t:NWgen}
(Nisan-Wigderson pseudo random generator.)
For every constants $c$ and $d$ there exist a constant $c'$ with the following property. There exists a function $\mbox{NW-gen}:\zo^{O(\log^{c'}n)} \mapping \zo^n$ such that for any circuit $G$ of size $2^{(\log n)^c}$ and depth $d$,
\[
 | \prob_{s \in \zo^{O(\log^{c'}n)}}[G(\mbox{NW-gen}(s)) = 1] - \prob_{z \in \zo^n} [G(z)=1]| < 1/100.
\]
Moreover, there is a procedure that on inputs $(n, i,s)$ produces the $i$-th bit of $\mbox{NW-gen}(s)$ in time
$\poly(\log n)$.

\end{theorem}

\section{Upper bound for the randomness rate of single source Kolmogorov extractors with constant nonuniformity}
In this section we show the limitations of what quality of randomness can be extracted with a bounded quantity of advice.
We obtain this as a consequence of a result of Vereshchagin and Vyugin~\cite{ver-vyu:j:kolm}. To state their result, let us fix $n = $ length of the source, $h =$ number of bits of advice that is allowed, and $m = $ the number of extracted bits. Let  $H = 2^{h+1}-1$.
\smallskip

\begin{theorem}[\cite{ver-vyu:j:kolm}]
\label{t:vervyu}
There exists a string $x \in \zo^n$  with $C(x) > n- H \log (2^m+1) \approx n - Hm$ such that any string $z \in \zo^m$ with $C(z \mid x) \leq h$ has complexity $C(z) < h + \log n + \log m + O(\log \log n, \log \log m)$. 
\end{theorem}
\smallskip

The next theorem, a consequence of Theorem~\ref{t:vervyu}, shows that no Kolmogorov extractor for sources with randomness rate $\sigma$ and that uses $h$ bits of advice about the source can output strings with randomness rate larger than $1 - (1-~\sigma)/H$.
\smallskip

\begin{theorem}
\label{t:limits}
Assume that the parameters $m, h, \sigma$ are computable from $n$ and satisfy the following relations: $0 < \sigma < 1, h > 0, 0 < m < n$, $m = \omega (\log n + h)$. 

Let $f: \zo^n \times \zo^h \mapping \zo^m$ be a computable ensemble of functions such that for every $ x \in \zo^n$ with $C(x) \geq \sigma \cdot n$, there exists a string $\alpha_x$ such that $C(f(x , \alpha_x)) \geq (1-\epsilon)\cdot m$. Then $\epsilon \geq \frac{1-\sigma}{H}- o(1)$.
\end{theorem}
\smallskip

\begin{proof} Let $m' = \min (\lfloor \frac{1-\sigma}{H} \cdot n\rfloor, m)$. Note that $m' \geq  \lfloor \frac{1-\sigma}{H} \cdot m\rfloor$.
Let $x$ be the string guaranteed by the Vereshchagin-Vyugin Theorem~\ref{t:vervyu} for the parameters $n, h + c, m'$, where $c$ is a constant that will be specified later. Note that $C(x) > n - H\cdot m' \geq \sigma \cdot n$. By assumption there is a string $\alpha_x$ such that $C(f(x, \alpha_x)) \geq (1-\epsilon)m$. Let $z$ be the prefix of length $m'$ of $f(x,\alpha_x)$. Note that $C(f(x, \alpha_x)) \leq C(z) + (m-m') + 2\log m + O(1)$, which implies that
$C(z) \geq (1-\epsilon)m - m + m' - 2 \log m - O(1) \geq \frac{(1-\sigma)m}{H} - \epsilon m - 2 \log m - O(1)$.

We also have $C(z \mid x) \leq |\alpha_x| + c = h + c$, for some constant $c$. It follows from Theorem~\ref{t:vervyu} that
$C(z) < h + \log n + \log m' + O(\log \log n, \log \log m')$. So,
$\frac{(1-\sigma)}{H}m - \epsilon m - 2 \log m - O(1) \leq h + \log n + O(\log \log n, \log \log m')$, which implies that $\epsilon \geq \frac{1-\sigma}{H} - \frac{h + O(\log n)}{m} = \frac{1-\sigma}{H} - o(1)$.
\qed \end{proof}

\section{Single source Kolmogorov extractors with small advice}
We move to showing the positive results in Fact~\ref{t:boundsinformal} and Fact~\ref{t:boundspolysizeinformal} regarding randomness extraction with small advice that complement the negative result in Theorem~\ref{t:limits}. The constructions use the parameters $n, n_1, m, k, \delta$ and $d$. We denote $N=2^n, N_1 = 2^{n_1}, M=2^m, \Delta = 2^{\delta}$ and $D= 2^{d}$. We identify in the natural way a function $E: \zo^n \times \zo^{n_1} \mapping \zo^m$ with an $[N] \times [N_1]$ table colored with colors from $[M]$. For $A \subseteq [M]$, we say that an $(u,v)$ cell of the table is an $A$-cell if $E(u,v) \in A$. The reader might find helpful to consult the proof plan presented in the Introduction. As explained there the notion of a \emph{balanced table} plays an important role.
\begin{definition}
A table $E: [N] \times [N_1] \mapping [M]$ is $(K,D,\Delta)$-balanced if for any $B \subseteq [N]$ with $|B| \geq K$, for any $A \subseteq [M]$ with $\frac{|A|}{M} \geq \frac{1}{D}$, it holds that
\[
\frac{|\mbox{$A$-cells in $B \times [N_1]$}|}{|B| \times N_1} \leq \Delta \cdot \frac{|A|}{M}.
\]
\end{definition}
The following lemma shows that a balanced table is a good Kolmogorov extractor.
 \begin{lemma}
 \label{l:tableext}
  Let $E: [N] \times [N_1] \mapping [M]$
be a $(K,D,\Delta)$-balanced table and $d = \delta + O(1)$. Suppose $C(E \mid n) = O(1)$ and $n_1, k, d$, and $\delta$ are computable from $n$. Let $(x,y) \in [N]\times [N_1]$ be such that $C(x \mid n) \geq k + O(1)$ and $C(y \mid x ) \geq n_1$. Let $z = E(x,y)$. Then $C(z \mid m) > m-d$. 

(Note: $O(1)$ means that there exist constants, depending only on the universal machine, for which the statements hold.)
  \end{lemma}
 \begin{proof}  Suppose $C(z \mid m) \leq m-d$. 
 
 Let $A = \{w \in \zo^m \mid C(w \mid m) \leq m-d + O(1)\}$, where the constant $O(1)$ is chosen so that $|A| \geq 2^{m-d}$. Also note that $|A| \leq 2^{m-d+O(1)}$. 
 
 We say that a row $v$ is \emph{bad} if the number $A$-cells in the $\{v \} \times [N_1]$ rectangle of $E$ is $> \Delta \cdot \frac{|A|}{M} \cdot N_1$.
 The number of bad rows is at most $K$, because the table $E$ is $(K,D,\Delta)$-balanced.
 Therefore a bad row $v$ is described by the information needed to enumerate the bad rows (and this information is derivable from $n$) and from its rank in the enumeration of bad rows. So, if $v$ is bad, $C(v \mid n) < k+O(1)$. 
 
 Since $C(x \mid n) > k + O(1)$, it follows that $x$ is good. Therefore, the number of $A$-cells in the $\{x\} \times [N_1]$ rectangle of $E$ is $\leq \Delta \cdot \frac{|A|}{M} \cdot N_1 =2^{\delta - d + n_1 + O(1)}$.
 
 Note that, by our assumption, the cell $(x,y)$ is an $A$-cell. Since we can effectively enumerate the $A$-cells in the table $E$, the string $y$, given $x$, can be described by the rank of $(x,y)$ among the $A$-cells in the $\{ x \} \times [N_1]$ rectangle of $E$.
 
 So, $C(y \mid x) \leq \delta - d + n_1 +O(1)$ and the right hand side is less than $n_1$ for an appropriate choice of the constant $O(1)$ in the relation between $d$ and $\delta$. We obtain that $C(y \mid x) < n_1$, contradiction.
 \qed \end{proof}
 \smallskip
 
 The next lemma establishes the parameters for which balanced tables exist.
 \begin{lemma}
 \label{l:probtable}
 Suppose the parameters satisfy the following relations: $D = O(\Delta), n/\delta = o(N_1)$, and $M = o(\delta\cdot K \cdot N_1)$. Then there exists a table $E:[N] \times [N_1] \mapping [M]$ that is $(K,D,\Delta)$-balanced.
 \end{lemma}
 \begin{proof}  The proof is by the probabilistic method and is presented in the Appendix.
 \qed \end{proof}
 \smallskip
 
 We can now prove Fact~\ref{t:boundsinformal}. The formal statement is as follows.
 \begin{theorem}
 \label{t:extractsmalladvice}
 Parameters: Let $m(n)$ and $h(n)$ be computable functions such that $m(n) < n$ for all $n$ and $h(n) = \omega (\log \frac{n}{m(n)})$.

 There exists a computable function $E:\zo^n \times \zo^{h(n)} \mapping \zo^{m(n)}$, such that for every $x \in \zo^n$ with $C(x \mid n) \geq m(n)$, there exists $\alpha_x \in \zo^{h(n)}$ such that
 $C(E(x, \alpha_x) | m(n)
 ) \geq m(n) - o(m(n))$.
 
 \end{theorem}
 \begin{proof} We take $\delta = \frac{n}{2^{0.5 h(n)}}$, $d = \delta + c$, where $c$ is the constant from Lemma~\ref{l:tableext}, $n_1 = h(n)$.
 
 By Lemma~\ref{l:probtable}, there exists a table $E: [N] \times [N_1] \mapping [M]$ that is
 $(K,D,\Delta)$-balanced, and by brute force one can build such a table from $n$. Thus we obtain such a table $E$ with $C(E \mid n) = O(1)$. We take $\alpha_x$ to be a string in $\zo^{h(n)}$ such that $C(\alpha_x \mid x) \geq h(n)$. Using Lemma~\ref{l:tableext}, we obtain that $C(E(x, \alpha_x) \mid m) \geq m(n) - d = m(n) -  \frac{n}{2^{0.5 h(n)}} - c =  m- o(m)$.
 \qed \end{proof}
 \smallskip
 
 Our next goal is to derandomize the construction in Theorem~\ref{t:extractsmalladvice}. As explained in the Introduction the key observation is that checking if a table is balanced can be done, in an approximate sense, by constant-depth circuits with relatively small size.
 \smallskip
 
 \begin{lemma}
 \label{l:countcircuit}
 The parameters $n_1,m,k, d$, and $\delta$ are positive integers computable from $n$  in polynomial time.
 We assume $k \leq n, m \leq k, d \leq n$.
 
 There exists a circuit $G$ of size $\poly(N^K)$ and constant depth such that for any table $E: [N] \times [N_1] \mapping [M]$,
 
 (a) if $G(E) = 1$, then $E$ is $(K,D,1.03 \Delta)$-balanced,
 
 (b) if $E$ is $(K,D, \Delta)$-balanced, then $G(E) = 1$.
 
 \end{lemma}
 \smallskip
 
 \begin{proof} Let $a = (1/0.99) \Delta \cdot 1/D \cdot K \cdot N_1$. Let us fix for the moment  a set
 of rows $B \subseteq [N]$ of size $|B| = K$ and a set of colors $A \subseteq [M]$ of size $|A| = M/D$. Let $x_{B,A}$ be a binary string indicating which cells in the $B \times [N_1]$ rectangle of $E$ are $A$-colored. Formally, $x_{B,A}$ is the string of length $K \cdot N_1$, whose $\langle i, j \rangle$-th bit is $1$ if the cell $(i,j)$ in the rectangle $B \times [N_1]$ of $E$ is an $A$-cell and $0$ if it is not.

 By Ajtai's Theorem~\ref{t:ajtai}, there exists a polynomial-size constant-depth circuit $G'$ (which does not depend on $B$ and $A$) with $a$ hardwired and such that 
 \begin{itemize}
	\item  $G'(x_{B,A}) = 1$ if the number of $A$-cells in $B \times [N_1]$ is at most $(1-0.01) \cdot a$, and
	\item $G'(x_{B,A}) = 0$ if the number of $A$-cells in $B \times [N_1]$ is at least $(1+0.01) \cdot a$.
\end{itemize}
Now we describe the circuit $G$.

 The circuit $G$ on input an encoding of the table $E$ (having length $N \cdot N_1 \cdot m$) computes in constant depth a string $x_{B,A}$  for every $B \subseteq [N]$ with $|B| = K$ and for every $A \subseteq [M]$ with $A = M/D$. There are ${N \choose K} {M \choose M/D} = \poly(N^K)$ such strings $X_{B,A}$. Each such string $x_{B,A}$ is the input of a copy of $G'$. The output gates of all the copies of $G'$ are connected to an AND gate, which is the output gate.
 
 If $G(E) = 1$, then $G'(x_{B,A}) = 1$ for all $B$'s and $A$'s as above. This implies that for all $B \subseteq [N]$ with $|B| \geq K$ and all $A \subseteq [M]$ of size $\geq M/D$, the number of $A$-cells in the $B \times [N_1]$ rectangle of $E$ is at most $(1+0.01)a \leq (1.03) \cdot \Delta \cdot (1/D) \cdot K \cdot N_1$, \ie, $E$ is $(K, D, 1.03\Delta)$-balanced.
 
 In the other direction, if $E$ is $(K,D, \Delta)$-balanced then for all $B \subseteq [N]$ with $|B| = K$ and for all $A \subseteq [M]$ with $A = M/D$, the number of $A$-cells in $B \times [N_1]$ is at most $\Delta \cdot (1/D) \cdot K \cdot N_1 = (1-0.01)a$, which implies that $G(E) =1$.
\qed \end{proof}
 \smallskip
 
 We next prove Fact~\ref{t:boundspolysizeinformal}. The formal statement is as follows.
\smallskip

 \begin{theorem}
 \label{t:derandextract}
 Parameters: Let $m(n)$ and $h(n)$ be polynomial-computable functions such that $m(n) \leq \log^c n$ for some constant $c$ and
 $h(n) = \omega (\log(\frac{n}{m(n)}))$. 
 
 There exists a function $E: \zo^n \times \zo^{h(n)} \mapping \zo^{m(n)}$,  computable by an  effectively constructible circuit having polynomial size and the following property: For every $x \in \zo^n$ with $C(x \mid n) \geq m(n) + O(1)$, there exists a string $\alpha_x \in \zo^{h(n)}$ such that $C( E(x, \alpha_x) \mid m(n)) \geq m - o(m)$.
 \end{theorem}
 \smallskip

 \begin{proof} Let $k = m(n), \delta = \frac{n}{2^{0.5 h(n)}}, d = \delta + c + \log 1.03$  (where $c$ is the constant from Lemma~\ref{l:tableext}), and  $n_1 = h(n)$.
 
  Let $G$ be the circuit promised by Lemma~\ref{l:countcircuit} for these parameters.  Let $d_{Ajtai}$ be the depth of the circuit $G$.  Note that the size of $G$ is bounded by $\poly(N^K) < 2^{ \log^{c+2} N}$ (taking into account the bound on $m$ given in the hypothesis).
  
  Let $\tilde{N} = N \cdot N_1 \cdot m$. This is the size of an encoding of a table $E: [N] \times [N_1] \mapping [M]$. Let $\mbox{NW-gen} : \zo^{\log^{c'}(\tilde{N})} \mapping \zo^{\tilde{N}}$ be the Nisan-Wigderson pseudo-random generator given by Theorem~\ref{t:NWgen} that fools circuits of for depth $d_{Ajtai}$ and size $2^{ \log^{c+2} N}$ .

 The probabilistic argument in Lemma~\ref{l:probtable} can be modified to show that among the 
 tables of type $E: [N] \times [N_1] \mapping [M]$ the fraction of those which are $(K,D, \Delta)$-balanced is at least $0.51$.  Since $G$ accepts all such tables,
 \[
  \prob_{E \in \zo^{\tilde{N}}}[ G(E) = 1] \geq 0.51.
  \]
  Since the circuit $G$ has depth equal to $d_{Ajtai}$ and size bounded by $2^{ \log^{c+2} N}$, it follows that if we replace a random $E \in \zo^{\tilde{N}}$ by $\mbox{NW-gen}(s)$ for a random seed $s \in \zo^{\log^{c'}(\tilde{N})}$, we obtain
 \[
\prob_{s \in \zo^{\log^{c'}(\tilde{N})}}[ G(\mbox{NW-gen}(s)) = 1] \geq 0.5.
 \]
 We only need the fact that there exists a string $s \in \zo^{\log^{c'}(\tilde{N})}$
 such that $\mbox{NW-gen}(s)$ is a table $E: [N] \times [N_1] \mapping [M]$ that is $(K,D, 1.03 \Delta)$-balanced. We fix such an $s$ that is computable from $n$ (say, the smallest $s$ that has the property) and the corresponding table $E$ produced by the Nisan-Wigderson pseudo-random generator on seed $s$.
 
Let us consider $x \in \zo^n$ with $C(x \mid n) \geq k$ and $\alpha_x \in \zo^{n_1}$ with $C(\alpha_x \mid y)\geq n_1$. Since $E$ is $(K, D, 1.03\Delta)$-balanced, it follows from Lemma~\ref{l:tableext} that $C(E(x,\alpha_x)\mid m) \geq m - d = m - o(m)$.

Now, let us view $E$ (which is $\mbox{NW-gen}(s)$) as a function ${\rm E}: \zo^n \times \zo^{n_1} \mapping \zo^m$. From the properties (\ie, the ``Moreover ..." in Theorem~\ref{t:NWgen}) of the Nisan-Wigderson pseudo-random generator, it follows that this function can be computed by a polynomial-size circuit which has $s$ hardwired (note that the size of $s$ is $\poly(n)$). Since $s$ is also computable from $n$,  one can compute a description of the circuit, \ie, the circuit is effectively constructible. 
\qed \end{proof}

 \bibliography{c:/book-text/theory}

\bibliographystyle{alpha}

\newpage

\appendix
 
 \section{Appendix}
\medskip
\if01
 {\bf Proof of Claim~\ref{c:balancedt}.} 
 \smallskip
 
 We use the probabilistic method. It is enough to show the assertion for all $B_1$, $B_2$ and $A$ having sizes exactly $K_x$, $K_y$, and respectively $\frac{M}{D}$. Let us consider a random function $E: \zo^n \times \zo^n \mapping \zo^m$. Fix $B_1, B_2$ and $A$, satisfying the above requirement on their sizes. By the Chernoff's bound,
 \[
 \begin{array}{ll}
 \prob[|\mbox{$A$-cells in $B_1 \times B_2$}| \geq 2 \cdot \frac{|A|}{M} \cdot |B_1| \cdot |B_2| ] & \\
 \quad\quad\quad\quad \leq e^{-(1/3) (|A|/M)|B_1||B_2|} = e^{-(1/3)(1/D)K_x K_y}. &
 \end{array}
 \]
 The sets $B_1$, $B_2$, and $A$ can be chosen in ${N \choose K_x} \cdot {N \choose K_y} \cdot {M \choose M/D} \leq N^{K_x} \cdot N^{K_y} \cdot (eD)^{M/D} = e^{K_x \ln N + K_y \ln N + (M/D) + (M/D) \ln D}$ ways.
 Since $t_x \geq 13 \log n + I(x : y) + O(1)$  and $d < 6 \log n + I(x : y) + O(1)$, it follows that $t_x \geq d + 7 \log n +O(1)$ and from here $k_x \geq d + \log n + O(1)$. Since $t_y \geq 7 \log n + I(x : y) + O(1)$, it follows that $k_y \geq d + \log n  + O(1)$. It can be easily checked that, given these bounds for $k_x$ and $k_y$ and for an appropriate choice of the constant in the definition of $d$,
 \[
 e^{-(1/3)(1/D)K_x K_y} \cdot e^{K_x \ln N + K_y \ln N + (M/D) + (M/D) \ln D} < 1.
 \]
 Thus, the probability that a random $E$ satisfies the requirements is less than $1$, which implies that there exists an $E$ satisfying the claim.~\qed
 \medskip
 \fi
 {\bf Proof of Lemma~\ref{l:probtable}   .}
 \smallskip
 
  The proof is by the probabilistic method. Consider a random function $E:[N]\times [N_1] \mapping [M]$.  We evaluate the probability that $E$ fails to be $(K, D, \Delta)$-balanced. Note that if $E$ fails to be $(K,D,\Delta)$-balanced, then there exists a set $B \subseteq [N]$ of size exactly $K$ and a set $A \subseteq [M]$ of size exactly $M/D$ such that the fraction of $A$ cells in the $B \times [N_1]$ rectangle of $E$ is greater  than $\Delta \cdot |A|/M$. Let us call this latter event ${\cal S}$. We show that the probability of  ${\cal S}$ is less than $1$.  Fix 
 $B \subseteq [N]$ of size $K$ and $A \subseteq [M]$ of size $M/D$. For a fixed $(x,y) \in B \times [N_1]$, $\prob[E(x,y) \in A]=|A|/M$. The expected number of $A$-cells in $B \times [N_1]$ is $\mu = |B| \cdot N_1 \cdot |A|/M$. Let $\Delta'=\Delta -1$. 
 
 We use the following version of the Chernoff bound. If $X$ is a sum of independent Bernoulli random variables, and the expected value $E[X]= \mu$, then
$\prob[X \geq (1+\Delta)\mu] \leq e^{-\Delta (\ln (\Delta/3)) \mu}$.\footnote{The standard Chernoff inequality $\prob(X \geq (1+\Delta) \mu] \leq \big( \frac{e^\Delta}{(1+\Delta)^{(1+\Delta)}}\big)^\mu$ is presented in many textbooks. It can be checked easily that $\frac{e^\Delta}{(1+\Delta)^{(1+\Delta)}} < e^{-\Delta \ln (\Delta/3)}$.}
 
 Using these Chernoff bounds, 
 \[
 \prob [|\mbox{$A$-cells in $B \times [N_1]|$} > (1+\Delta')\mu]\leq e^{-\Delta' (\ln (\Delta'/3)) \mu}.
 \]
 The set $B$ can be chosen in ${N \choose K} \leq N^K$ ways. The set $A$ can be chosen in ${M \choose M/D} \leq (eD)^{M/D}$ ways.  It follows that the probability of ${\cal S}$ is bounded by
 \[
 N^K \cdot (eD)^{M/D} \cdot e^{-\Delta' (\ln (\Delta'/3)) \cdot K \cdot N_1 \cdot (1/D)},
 \]
 which, taking into account the relations between parameters, is less than $1$.~\qed

\end{document}